\documentclass[10pt,twoside]{article}

\usepackage{amsfonts,amsmath,amssymb,amsthm}
\usepackage[dvips,final]{epsfig}
\usepackage[dvips]{graphicx}
\newcommand{\R}{\mathbb{R}}

\newtheorem{theorem}{Theorem}
\newtheorem{proposition}{Proposition}
\newtheorem{corollary}{Corollary}

\newtheorem{definition}{Definition}
\makeatletter
\evensidemargin \oddsidemargin
\makeatother

\begin{document}
\thispagestyle{empty}
\setcounter{page}{1}

\noindent


\begin{center}
{\large\bf Symmetric Relative Equilibria in the Four-Vortex Problem with Three Equal Vorticities}

\vskip.20in

Ernesto P\'erez-Chavela$^{1}$,  Manuele Santoprete$^{2}$ \ and\ Claudia Tamayo $^{1}$ \\[2mm]
{\footnotesize
$^{1}$Departamento de Matem\'aticas\\
Universidad Aut\'onoma Metropolitana-Iztapalapa. M\'exico, D.F., M\'exico\\[5pt]
$^{2}$Department of Mathematics\\
Wilfrid Laurier University. Waterloo, Ontario, Canada\\
}
\end{center}

{\footnotesize
\noindent
{\bf Abstract.}  We examine in detail the relative equilibria of the 4-vortex problem when three vortices have equal strength, that is, $\Gamma_{1} = \Gamma_{2} = \Gamma_{3} = 1$, and $\Gamma_{4}$ is a real parameter.  We give the exact number of  relativa equilibria and bifurcation values.    We also study the relative equilibria in the vortex rhombus problem. \\[3pt]
}

\vskip.2in


\section{Introduction}

The origins of vortex dynamics lie in the famous work of  Helmholtz of  1858 , where he introduced the concepts of vortex line, vortex filament and derived the vorticity equation for an ideal incomprenssible fluid \cite{H}.   Helmholtz also introduced planar point vortices and their equations of motion in order to model a $2$-dimensional slice of columnar vortex filaments.  Some years later, Kirchhoff (1876)  gave a Hamiltonian formulation of Helmholtz's equations for point vortices (see \cite{K} for more details).  This model has been widely used to provide finite-dimensional approximation to vorticity evolution in fluid dynamics.  Kirchhoff proved that $n$ point vortices in the plane located at $z_i=(x_i, y_i) \in \R^{2}$ with vortex strength $\Gamma_i \neq 0 \in \R$ for $i=1, \dots, n$ satisfy
 \begin{equation}\label{Kirchof}
\Gamma_{i} \frac{dx_{i}}{dt} = \frac{\partial H}{\partial y_{i}}, \qquad \Gamma_{i} \frac{dy_{i}}{dt} = -\frac{\partial H}{\partial x_{i}},
\end{equation}
where
\[
H = -\frac{1}{2}\sum_{i<j} \Gamma_{i}\Gamma_{j} \log [(x_{i} - x_{j})^{2} + (y_{i} - y_{j})^{2}].
\]

Computing the derivatives indicated in  (\ref{Kirchof}), we can find the velocity of the $i$-th vortex:
\begin{equation}\label{Hamiltoniano}
\frac{dx_{i}}{dt} = -\sum_{j\neq i} \Gamma_{j}\frac{y_{i} - y_{j}}{r_{ij}^{2}}, \qquad  \frac{dy_{i}}{dt} = \sum_{j\neq i} \Gamma_{j}\frac{x_{i} - x_{j}}{r_{ij}^{2}},
\end{equation}
where
\[
r_{ij}^{2} = (x_{i} - x_{j})^{2} + (y_{i} - y_{j})^{2}.
\] 

Let $
\textbf{J} =
\left(
\begin{array}{cc}
0  &  \textbf{I} \\
-\textbf{I} & 0
\end{array}
\right),
$ the standard symplectic matrix $2n \times 2n$, and let $\nabla_j$ denote the two-dimensional partial gradient with respect to $z_j$, the previous equation can be written in vectorial form as
\begin{equation}\label{em}
\textbf{G} \dot{z_{i}} = \textbf{J} \nabla_{i} H = -J \sum_{j=1}^{n}\frac{\Gamma_{i}\Gamma_{j}}{r_{ij}}(z_{i}-z_{j}), \qquad 1 \leq i \leq n,
\end{equation}
where $\textbf{G} = \; \text{diag}(\Gamma_{1}, \cdots ,\Gamma_{n})$ is the diagonal matrix.

In general the $n$-vortex system is simpler than the $n$-body problem of point masses governed by Newtonian gravity for a given $n$.  For example the three-vortex system is always integrable, whereas in the Newtonian three-body problem there are chaotic regimes.  Somewhat offsetting this relative simplicity is the larger set of parameter values that we must investigate  (since $\Gamma_i<0$ is allowed).

In 2001, Kossin and Schubart  \cite{KS} conducted numerical experiments describing the evolution of thin annular rings with large vorticity as a model for the behavior seen in the eyeball of intensifying hurricanes.  In a conservative, idealized setting, they find examples of ``vortex crystals'', formations of mesovortices (namely a vertical vortex of air associated with a thunderstorm that occurs at less  one kilometer from the ground) that rigidly rotate as a solid body. One particular formation of four vortices, situated very close to a rhombus configuration, is observed to last for the final 18 hours of a 24-hour simulation.  Rigidly rotating polygonal configurations have also been found in the eyeballs of hurricanes in weather research and forecasting models from the Hurricane Group at the National Center for Atmospheric Research  (see the website \cite{C} for some revealing simulations).

It is natural to explore these rigidly rotating configurations in a dynamical systems setting by studying relative equilibria of the planar $n$- vortex problem.  

The three-vortex case was extensively studied by Gr\"{o}bli \cite{G}, Kossin and Schubart \cite{KS} and Hern\'andez-Gardu\~{n}o and Lacomba \cite{GL}.  Equilateral triangles are always relative equilibria.  There are also collinear equilibria and/or relative equilibria, which are determined by a cubic equation in a shape parameter with coefficients which are linear in the vortex strengths.  Depending of the parameters it is possible to obtain one, two or three collinear relative equilibria. In 2008 M. Hampton and R. Moeckel proved the finiteness of relative equilibria in the four vortex problem \cite{HM}. The four vortex case with two pairs of equal vorticities was studied by Hampton, Roberts and Santoprete \cite{HS}. The  linear and nonlinear stability of certain symmetric configurations of point vortices on the sphere  forming relative equilibria  was studied in \cite{LMR}. These configurations include, in particular, kite configurations. 

The purpose of this paper is to study the planar relative equilibria of the 4-vortex problem with three equal vorticities.   We set three vorticities equal to 1, and the fourth vorticity $\Gamma_{4}$ is taken as a real parameter.   Our main goal is to classify, describe, and count the number and type of solutions of relative equilibria as $\Gamma_{4}$ varies. We say that a planar non-collinear central configuration of the 4-vortex problem has a kite shape if it has an axis of symmetry passing through two of the vorticities.  The kite configuration is \textit{convex} if none of the bodies is located in the interior of the convex hull of the other three, otherwise and if the configuration is not collinear we say that the kite configuration is \textit{concave}.  When counting solutions, we use the standard convention from celestial mechanics that solutions which are identical under scaling or rotation are considered equivalent.  

When $\Gamma_{4} =1$  (all vortex strengths equal), there are 34 solutions.  In this case, it is not possible to distinguish if the vorticity $\Gamma_{4}$ is at a vertex of the triangle formed by the convex hull of the four vorticities, or if it is in its interior.  There are only four geometrically distinct configurations: a square, an equilateral triangle with a vortex at the center, an equilateral triangle with a vortex on a vertex of the triangle formed by the convex hull of the four vorticities, and a collinear configuration.  This is different from the Newtonian case, where an additional symmetric, concave solution exists, consisting of an isosceles triangle with an interior body on the axis of symmetry.  An interesting bifurcation occurs as $\Gamma_{4}$ decreases through $\Gamma_{4} =1$, the equilateral triangle with the vortex 4  at the center of the triangle goes to an isosceles triangle with the interior vortex on the line of symmetry, and the equilateral triangle with the vortex 4 at the vertex of the equilateral triangle formed by the convex hull of the other three vorticities vanish.  Thus, the number of solutions decreases from 34 to 23 for the case $\Gamma_{4}>1$.  If $\Gamma_{4}$ increases from $0$ through $1$, the equilateral triangle splits into three different solutions.  If the vortex 4 is at the vertex of the equilateral triangle formed by the convex hull of the other three vorticities, then the solution for $\Gamma_{4} = 1$ bifurcates into two different isosceles triangles.  If vortex 4 is at the center of the triangle, the equilateral solution goes to an isosceles configuration.   Thus, the number of solutions decreases  from 34 to 29 for the case $0<\Gamma_{4}<1$.

As $\Gamma_{4}$ flips sign, there is one bifurcation value at $\Gamma_{4} = 0$.  When $\Gamma_{4} = 0$ there are 26 solutions.  The equilateral triangle bifurcates in two different solutions, a isosceles triangle with $\Gamma_{4}$ in its interior, and  $\Gamma_{4}$ on the vertex of an isosceles triangle.  If $\Gamma_{4}$ approach to 0 there is one kite concave configuration having $\Gamma_{4}$ in the interior of an isosceles triangle.  In the collinear case we have bifurcation value at $\Gamma_{4} = -1/2$.

The paper is organized as follows:  In Section 2, we define a relative equilibrium and explain how to use mutual distances as variables in the 4-vortex problem.  In Section 3, we describe the relevant algebraic techniques used to analyze and quantify the number of solutions.  Section 4 examines the interplay between symmetry and equality of vorticities in two special cases: absolute equilibria and rigid translations.  Sections 5 and 6 cover the collinear case and kite configurations.

\section{Relative Equilibria}

A motion of $n$ vortices is said to be a \textit{relative equilibrium} if, and only if there exists a real number $\lambda$, called angular velocity, such that, for every $i,j$ and for all time $t$:
\[
z_{i} - z_{j} = e^{-J\lambda t}(z_{i}(0) - z_{j}(0)).
\]
Then one of the following statements is satisfied:

\begin{itemize}
\item If the $z_{j}$ are constant, then the motion is said to be an \textit{absolute equilibrium}.  In this case, we have $\lambda = 0$.

\item If there exists a velocity of translation $v \neq 0$ such that, for every $j$ and for all time $t$:  $z_{j}(t) = z_{j}(0) + tv$, then the motion is said to be a \textit{rigid translation}.  Again, we have $\lambda = 0$.

\item If $\lambda \neq 0$. then there exists a \textit{center of rotation} $c$ such that, for every $j$ and for all time $t$:
\[
z_{j}(t) = c + e^{-J\lambda t}(z_{j}(0) -c).
\]
If, moreover: $\sum_{j=1}^{n}\Gamma_{j}\neq 0$, the center of rotation is the center o vorticity.
\end{itemize}

In fact, looking for these motions we observe that the search of relative equilibria is equivalent to look for the configurations which generate them, as shown in the following proposition which is proved in \cite{O}:

\begin{proposition}\label{mer}
A motion of $n$ vortices is a relative equilibrium if, and only if, at a certain time, there exists a real number $\lambda$ such that for every $i, j$ we have $ \dot{z_{i}}-\dot{z_{j}} = - \lambda(z_{i} - z_{j})$.  
\begin{itemize}
\item[1.]  It is an absolute equilibrium if, and only if, at a certain time,  for every $j$ we have $ \dot{z}_{j} = 0.$

\item[2.]  It is a rigid translation if, and only if, at a certain time, there exists $v \neq 0$ such that for every $j$ we have $ \dot{z}_{j} = v$.

\item[3.]  It is a relative equilibrium with $\lambda \neq 0$ if, and only if, at a certain time, there exists $c$ such that  for every $j$ we have, 
\end{itemize} 
\begin{equation}\label{re}
\dot{z}_{j} = -\lambda(z_{j} - c).
\end{equation}
\end{proposition} 

\begin{definition} 
The following quantities are defined:
\begin{align*}
\text{Total vorticity} & \; & \Gamma &= \sum \Gamma_{l}\\
\text{Angular momentum} & \; &  L &= \sum_{l<k}\Gamma_{l}\Gamma_{k}\\
\text{Moment of vorticity} & \; &  M &= \sum \Gamma_{l}z_{l}\\
\text{Center of vorticity } & \; &  c &= M/ \Gamma  \;(\text{when} \; \Gamma \neq 0)\\
\text{Moment of inertia} & \; &  I &= \frac{1}{2}\sum \Gamma_{l}\|z_{l}\|^{2}
\end{align*}
\end{definition}

Applying the matrix $\textbf{G}$ to both sides of equation (\ref{re} )  in Proposition (\ref{mer}), and using the equations of motion (\ref{em}) we obtain
\begin{equation}\label{inercia}
\lambda \nabla (I-I_{0}) = \nabla H,
\end{equation}
where $\nabla = (\nabla_{1}, \dots , \nabla_{n})$ and $I = I_{0}$. We observe that equation (\ref{inercia}) is Lagrange multiplier problem, with $\lambda$ as the Lagrange multiplier and any solution can be interpreted as a critical point of the Hamiltonian $H(z)$ under the condition that $I$ remains constant. Using the homogeneity of the functions $H$ and $I$, equation (\ref{inercia}) implies that the angular velocity $\lambda$ in a relative equilibria is given by
\begin{equation}\label{lambda}
\lambda = \frac{-L}{2I}.
\end{equation}  

\subsection{Equations in mutual distances}

We now consider the case of $n=4$ vortices.  Our presentation follows the approach of \cite{Sch}   in describing the work of Dziobek \cite{Dz} for the Newtonian $n-$body problem.  We want to express equation (\ref{inercia}) in terms of the mutual distance variables $r_{ij}$.  In 1900, Dziobek gave the innovative idea of using the mutual distances of the bodies as unknowns in order to formulate the equations of motion.   His work reduces the problem of searching for relative equilibria to the study of systems of non-linear polynomial equations that exhibit many symmetries.

Between four vortices there are six mutual distances, which are not independent if the vortices are planar;  generically they describe a tetrahedron in $\R^{3}$ in place of a configuration in the plane. In order that they describe a planar relative equilibria we need an additional constraint, which is obtained by setting the volume of the tetrahedron equal to zero.  This restriction follows from the Cayley-Menger determinant:

\[
S = \left|
\begin{array}{ccccc}
0 & 1& 1 & 1 & 1 \\
1 & 0 & r_{12}^{2} & r_{13}^{2} & r_{14}^{2} \\
1 & r_{12}^{2} & 0 & r_{23}^{2} & r_{24}^{2} \\
1 & r_{13}^{2} & r_{23}^{2} & 0 & r_{34}^{2} \\
1 & r_{14}^{2} & r_{24}^{2} & r_{34}^{2} & 0 \\
\end{array}
\right|.
\]

Hence, relative equilibria configurations are obtained as critical points of following equation
\begin{equation}\label{cp}
H - \lambda(I - I_{0}) - \frac{\mu}{32} S = 0,
\end{equation}
depending on $\lambda, \mu, r_{12}, \dots , r_{34}$, where $\lambda$ and $\mu$ are Lagrange multipliers.

\smallskip
\noindent
To find $S$ restricted to planar configurations, we use the following important formula
\begin{equation*}
\frac{\partial S}{\partial r_{ij}^{2}} = -32 A_{i}A_{j},
\end{equation*}
where $A_i$ is the oriented area of the triangle $T_i$ whose vertices are all except  the $i-$th vortex.  Setting the gradient of equation (\ref{cp}) equal to zero yields the equations
\[
\frac{\partial H}{\partial r_{ij}^{2}} - \lambda \frac{\partial I}{\partial r_{ij}^{2}} - \frac{\mu}{32}\frac{\partial S}{\partial r_{ij}^{2}} = 0.
\]

If $\Gamma \neq 0$,   $I$ can be written in terms of the mutual distances as 
\[
 I = \frac{1}{2\Gamma} \sum_{i<j}\Gamma_{i}\Gamma_{j} r_{ij}^{2},
\]
respectively, so
\[
\frac{\partial I}{\partial r_{ij}^{2}} = \frac{\Gamma_{i}\Gamma_{j}}{\Gamma}. 
\]

\noindent
Using this, we obtain the following equations for a four-vortex central configuration:
\[
\Gamma_{i}\Gamma_{j}(r_{ij}^{-2} + \lambda ') = \mu A_{i}A_{j},
\]
where $\lambda = \lambda ' \Gamma, I = I_{0}$ and $S = 0$. 
Explicity we have:
\begin{eqnarray}\label{md}
\Gamma_{1}\Gamma_{2}(r_{12}^{-2} + \lambda') = \mu A_{1}A_{2}, & \quad \Gamma_{3}\Gamma_{4}(r_{34}^{-2} + \lambda') = \mu A_{3}A_{4},\nonumber \\
\Gamma_{1}\Gamma_{3}(r_{13}^{-2} + \lambda') = \mu A_{1}A_{3}, & \quad \Gamma_{2}\Gamma_{4}(r_{24}^{-2} + \lambda') = \mu A_{2}A_{4},\\
\Gamma_{1}\Gamma_{4}(r_{14}^{-2} + \lambda') = \mu A_{1}A_{4}, & \quad \Gamma_{2}\Gamma_{3}(r_{23}^{-2} + \lambda') = \mu A_{2}A_{3}.\nonumber
\end{eqnarray}
This yields to the Dziobek equations for vortices \cite{Dz}:
\begin{equation}\label{Dziobek}
(r_{12}^{-2} + \lambda')(r_{34}^{-2} + \lambda') = (r_{13}^{-2} + \lambda')(r_{24}^{-2} + \lambda') = (r_{14}^{-2} + \lambda')(r_{23}^{-2} + \lambda').
\end{equation}
From the above equations we find 
\[
\lambda = \frac{r_{12}^{-2}r_{34}^{-2} - r_{13}^{-2}r_{24}^{-2}}{r_{12}^{-2} + r_{34}^{-2} - r_{13}^{-2} - r_{24}^{-2}} = \frac{r_{13}^{-2}r_{24}^{-2} - r_{14}^{-2}r_{23}^{-2}}{r_{13}^{-2} + r_{24}^{-2} - r_{14}^{-2} - r_{23}^{-2}} = \frac{r_{14}^{-2}r_{23}^{-2} - r_{12}^{-2}r_{34}^{-2}}{r_{14}^{-2} + r_{23}^{-2} - r_{12}^{-2} - r_{34}^{-2}}.
\]

Now, using the different ratios of two vorticities that can be found from the equations in (\ref{md}), we obtain:
\begin{eqnarray}\label{fv}
\frac{\Gamma_{1}A_{2}}{\Gamma_{2}A_{1}} = \frac{\rho_{23}+\lambda '}{\rho_{13} + \lambda '} = \frac{\rho_{24}+\lambda '}{\rho_{14} + \lambda '} = \frac{\rho_{23}-\rho_{24}}{\rho_{13}-\rho_{14}}, \nonumber \\
\frac{\Gamma_{1}A_{3}}{\Gamma_{3}A_{1}} = \frac{\rho_{23}+\lambda '}{\rho_{12} + \lambda '} = \frac{\rho_{34}+\lambda '}{\rho_{14} + \lambda '} = \frac{\rho_{23}-\rho_{34}}{\rho_{12}-\rho_{14}}, \nonumber \\
\frac{\Gamma_{1}A_{4}}{\Gamma_{4}A_{1}} = \frac{\rho_{24}+\lambda '}{\rho_{12} + \lambda '} = \frac{\rho_{34}+\lambda '}{\rho_{13} + \lambda '} = \frac{\rho_{24}-\rho_{34}}{\rho_{12}-\rho_{13}}, \nonumber \\
\frac{\Gamma_{2}A_{3}}{\Gamma_{3}A_{2}} = \frac{\rho_{13}+\lambda '}{\rho_{12} + \lambda '} = \frac{\rho_{34}+\lambda '}{\rho_{24} + \lambda '} = \frac{\rho_{13}-\rho_{34}}{\rho_{12}-\rho_{24}}, \\
\frac{\Gamma_{2}A_{4}}{\Gamma_{4}A_{2}} = \frac{\rho_{14}+\lambda '}{\rho_{12} + \lambda '} = \frac{\rho_{34}+\lambda '}{\rho_{23} + \lambda '} = \frac{\rho_{14}-\rho_{34}}{\rho_{12}-\rho_{23}}, \nonumber \\
\frac{\Gamma_{3}A_{4}}{\Gamma_{4}A_{3}} = \frac{\rho_{14}+\lambda '}{\rho_{13} + \lambda '} = \frac{\rho_{24}+\lambda '}{\rho_{23} + \lambda '} = \frac{\rho_{14}-\rho_{24}}{\rho_{13}-\rho_{23}}, \nonumber
\end{eqnarray}
where $\rho_{ij} = r_{ij}^{2}$, and $\lambda '$ is a constant.

Eliminating $\lambda'$ from equation (\ref{Dziobek}) and factoring we obtain the important relation
\begin{equation}\label{8}
(r_{13}^{2}-r_{12}^{2})(r_{23}^{2}-r_{34}^{2})(r_{24}^{2}-r_{14}^{2}) = (r_{12}^{2}-r_{14}^{2})(r_{24}^{2}-r_{34}^{2})(r_{13}^{2}-r_{23}^{2}).
\end{equation}

\noindent
Assuming that the six mutual distances determine a configuration in the plane, equations (\ref{8}) give a necessary condition for the existence of a four-vortex relative equilibrium.  The corresponding vortex strengths are then found from  equations  (\ref{fv}).

\section{Algebraic Techniques}

In this section we describe  an algebraic technique useful for analyzing solutions to our problem: elimination theory using Gr\"{o}bner bases.

\subsection{Elimination Theory and Gr\"obner Bases}

Elimination theory is the classical name for algorithmic approaches to eliminating some variables between polynomials of several variables.  The linear case would be handled by \textit{Gaussian elimination}.  In the same way, computational techniques for elimination can in practice be based on \textit{Gr\"obner bases} methods.

We mention some elements from elimination theory and the theory of Gr\"obner bases that will prove useful in our analysis.  For more details see \cite{DJD}.

Let $K$ be a field and consider the polynomial ring $K[x_1,...,x_n]$ of polynomials in $n$ variables $x_i$ over $K$. Let $f_1,...f_l$ be $l$ polynomials in $K[x_1,...,x_n]$ and consider the ideal $I = <f_1, . . . , f_l>=<\mathcal{F}>$ generated by these polynomials.

\begin{definition}
An order is a relation on the monomials $>$ such that it is a total order (ie, given two monomials, one is always greater than the other), such that if $x^\alpha>x^\beta$, and $x^\gamma$ any monomial, we have $x^{\alpha+\gamma}>x^{\beta+\gamma}$.   That is, it is preserved by multiplication of monomials. Finally, we want it to be a well-ordering. That is, every non-empty subset has a smallest element. 
\end{definition}

\begin{definition}
Let $I$ an ideal in $K[x_1, \dots , x_n]$.  We call the $k-$th elimination ideal the ideal $I \cap K[x_{k+1}, \dots ,x_n]$ in $K[x_{k+1}, \dots, x_n]$.
\end{definition}

\noindent
Note that if $k=0$, we just get $I$.

\smallskip

\begin{theorem}\textbf{(The Elimination Theorem)}
If $G$ is a Gr\"obner bases for $I$ with respect to lexicographic order with $x_1>x_2> \cdots >x_n$, then for all $0 \leq k \leq n$, we have
\[
G_k = G \cap K[x_k+1, \dots ,x_n]
\]
is a Gr\"obner bases for the $k-$th elimination ideal.
\end{theorem}

\section{Special Cases}

In this section we will study the \textit{equilibria} and \textit{rigid translation}.

A necessary condition on the vorticities for the existence of rigidly translating solutions is that $\Gamma = \sum \Gamma_{i} = 0$.  O'Neil  proved that for almost every such choice of vortex strengths, there are exactly $(n-1)!$ rigidly translating configurations \cite{O}.  He showed that for almost every choice of vortex strengths satisfying the necessary conditions $L = \sum \Gamma_{i}\Gamma_{j} =0$, there are $(n-2)!$ equilibria.

\noindent
In the four- vortex case, the equilibria can be found explicitly.  Since the equations are invariant under translations, set $z_{3} = (1, 0)$, and $z_{4} = (0,0)$.  The solutions for $(z_{1}, z_{2})$ are
\[
z_{1} = \frac{(2\Gamma_{4} + \Gamma_{2}, \pm  \Gamma_{2}\sqrt{3})}{2(\Gamma_{2} + \Gamma_{3} + \Gamma_{4})}, \quad z_{2} = \frac{(2\Gamma_{4} + \Gamma_{1}, \mp  \Gamma_{1}\sqrt{3})}{2(\Gamma_{1} + \Gamma_{3} + \Gamma_{4})}.
\] 

In our case, we suppose $\Gamma_{1} = \Gamma_{2} = \Gamma_{3} = 1$, $L=0$ when $\Gamma_{4} = -1$.  In this case, two solutions $(z_{1}, z_{2}, z_{3}, z_{4}) $ of relative equilibria are

\[
\left(\frac{(- 1, \sqrt{3})}{2}, \frac{(- 1, -\sqrt{3})}{2}, 1, 0\right), \quad \left(\frac{(- 1, - \sqrt{3})}{2}, \frac{(-1,  \sqrt{3})}{2}, 1, 0\right).
\]
\noindent
The configuration is an equilateral triangle with $\Gamma_{4}$ in the convex hull formed by the other vorticities.

\medskip

Now, for rigid translations we have that $\Gamma = 0$ when $\Gamma_{3} = -3$.  Following \cite{HM}, for relative equilibria we use the equations
\[
S_{1} = S_{2} = S_{3} = S_{4} = s_{0},
\]
and
\begin{equation}\label{17}
\frac{1}{s_{12}} + \frac{1}{s_{34}} = \frac{1}{s_{13}} + \frac{1}{s_{24}} = \frac{1}{s_{14}} + \frac{1}{s_{23}},
\end{equation}
where $S_{i} - \Gamma_{j} r_{ij}^{2}+\Gamma_{k}r_{ik}^{2} +\Gamma_{l} r_{il}^{2}$ and $I = \sum_{i=1}^{n}\|z_{i}-c\|^{2} = s_{0}$.  We clear denominators in the equations (\ref{17}) to get a polynomial system.  There are two types of symmetric relative equilibria.  The first is the equilateral triangle with vortex 4 at its center.  The second type is a concave kite with the three equal vorticities on the exterior isosceles triangle.

\section{Collinear Relative Equilibria}

Collinear relative equilibria of the four-vortex problem can be studied directly from equation (\ref{re}) since in this case it reduces to
 \begin{equation}\label{colineal}
 \lambda(x_{i} - c) = \sum_{i\neq j} \frac{\Gamma_{i}}{x_{i} - x_{j}},
 \end{equation}
where $c \in \R$.  Clearing denominators from these equations yields a polynomial system.  Rather than fixing $\lambda$ or $c$, we use the homogeneity and translation invariance of the system and set $x_{3} = -1, x_{4} = 1, \Gamma_{1} = \Gamma_{2} = \Gamma_{3} = 1$ and $\Gamma_{4}$ will be treated as a parameter. 

\medskip

Using this approach we get the following results:  

\subsection{Symmetric Solutions}

Given our setup, symmetric configurations correspond to solutions where $x_1=-x_2$.  In this case the center of vorticity $c$ is located at the origin.  Substituting these values in (\ref{colineal}) we get the following equation system:
\begin{eqnarray*}
  -2\lambda x_2^4+(2\lambda + 2\Gamma_4 +3)x_2^2 + 2(1-\Gamma_4)x_2 &=& 1,\\
  2\lambda x_2^4-(2\lambda + 2\Gamma_4 +3)x_2^2 + 2(1-\Gamma_4)x_2 &=& -1,\\
  (\Gamma_4-2\lambda)x_2^2 &=& \Gamma_4 +4 -2\lambda,\\
  (2\lambda-1)x_2^2 &=&2\lambda-5.
\end{eqnarray*}

\noindent
Solving the system we get that the only solutions are possible when all vorticities are equals
\[
-2\lambda x_2^4 + (2\lambda+2\Gamma_4+3)x_2^2-1 = 0, \quad \Gamma_4=1, \; \text{and} \; \lambda = \frac{x_2^2-5}{x_2^2-1}.
\]
Using $y=x_2^{2}$ we have four real solutions for $(x_1, x_2)=$
\[
 (-\sqrt{3} \mp \sqrt{2}, \sqrt{3} \pm \sqrt{2}) \quad \text{and} \quad (\sqrt{3} \pm \sqrt{2}, -\sqrt{3} \mp \sqrt{2}).
 \]

  \subsection{Asymmetric solutions}
 
 To locate any asymmetric solutions, we introduce the variables $u$ and $v$ along with the equations
 \[
 u(x_1+x_2)-1 \quad \text{and} \quad v(x_1-x_2)-1.
 \]
 
 \noindent
 Adding these two equations to the original polynomial system obtained from (\ref{colineal}), we compute a Gr\"obner basis $G_{col}$ with respect to the lex order where $c>\lambda>u>v>x_1>x_2>\Gamma_4$.  We get a 12th-degree polynomial in $x_2$ with coefficients in $\Gamma_4$.

\begin{equation*}
 \begin{split}
p(x_{2}) &= (13\Gamma^{4}+32\Gamma^{3}+2\Gamma^{5}+4+20
\Gamma +37\Gamma^{2}) {x_2}^{12}+ (-32\,{\Gamma }^{4}+20\,{\Gamma }^{2}+38\,\Gamma \\
& \quad -30{\Gamma }^{3} +12-8\,{\Gamma }^{5}) { x_2}^{11}- (1836\,{\Gamma }^{2}+312+
1250\,\Gamma +1204\,{\Gamma }^{3}+338\,{\Gamma }^{4}\\
& \quad +28\,{\Gamma }^{5}) {x_2}^{10} + (664\,{\Gamma }^{4}-1346\,\Gamma +88\,
{\Gamma }^{5}+1234\,{\Gamma }^{3}+100\,{\Gamma }^{2}-740)  {x_2}^{9}\\
& \quad + (3007\,{\Gamma }^{4}+13688\,{\Gamma }^{3} +26937\,{\Gamma }^{2}+254\,{\Gamma }^{5}+23290\,\Gamma +7020) {x_2}^{8}\\
& \quad + (-8636\,{\Gamma }^{3}-272\,{\Gamma }^{5}-5640\,{\Gamma }^{2}+8492\,\Gamma +8712-2656\,{\Gamma }^{4}) {x_2}^{7}\\
& \quad + (1288\,{\Gamma }^{5}+62688+156484\,\Gamma +145312\,{\Gamma }^{2} -14092\,{\Gamma }^{4}-62936\,{\Gamma }^{3}) {x_2}^{6}\\
& \quad+(-6092\,{\Gamma }^{3}+6476\,\Gamma -656\,{\Gamma }^{5}+9528\,{\Gamma }^{2}-4144\,{\Gamma }^{4} -5112) {x_2}^{5} \\
& \quad + (114080\,{\Gamma }^{3}+261207\,{\Gamma }^{2}+334552\,\Gamma +24859\,{\Gamma }^{4}+2078\,{\Gamma }^{5}+166860) {x_2}^{4}\\
& \quad+(-93138\,\Gamma +3864\,{\Gamma }^{5}-112340+91162\,{\Gamma }^{3}+32192\,{\Gamma }^{4}+78260\,{\Gamma }^{2}) {x_2}^{3}\\
& \quad + (2916\,{\Gamma }^{5}+13886\,{\Gamma }^{4}+9048+2348\,{\Gamma }^{3}-66484\,{\Gamma }^{2}-49626\,\Gamma) {x_2}^{2}\\
& \quad +(-12102\,{\Gamma }^{3} +7068+1080\,{\Gamma }^{5}+600\,{\Gamma }^{4}+9846\,\Gamma -6492\,{\Gamma }^{2}) {x_2}\\
& \quad-1148+162\,{\Gamma }^{5}+1050\,\Gamma -472\,{\Gamma }^{3}+1227\,{\Gamma }^{2}-711\,{\Gamma }^{4}.
 \end{split}
\end{equation*}

\begin{theorem}
Let the 4-vortex problem with $\Gamma_{1} = \Gamma_{2} = \Gamma_{3} = 1$ and $\Gamma_{4}$ taken as a parameter.  In addition, we fix $x_{3} = -1$ y $x_{4} = 1$.  There are collinear relative equilibria for every $\Gamma_{4} \in (-1, +\infty)$:
\begin{itemize}
\item If $\Gamma_{4} \in (-1, -1/2)$, we have six relative equilibria configurations.
\item If $\Gamma_{4} = -1/2$, we have seven collinear solutions.
\item If $\Gamma_{4} \in (-1/2, 1)$, there are twelve collinear solutions. 
\item If $\Gamma_4 = 1$, there are eight collinear solutions.
\item If $\Gamma_4 \in (1, +\infty)$, there are twelve collinear solutions.
\end{itemize}
\end{theorem}

Using \texttt{Mathematica}, we can find the $\Gamma_4$ values for which we have changes of sign in the polynomial $p(x_2)$.  By Descartes' rule of signs, we can find how  many real roots can take $p(x_2)$ and $p(-x_2)$.  Numerically we can use \textsl{Sturm}'s theorem to count the exact number of real roots, so that we obtain:  For $\Gamma \in [-\infty, -1]$, there are no real solutions other than the degenerate $x_{2} = \pm 1$ (i.e., the second vortex coincides with the fourth vortex).
For $\Gamma \in (-1, -1/2)$,  there are six different roots for $p(x_{2})$.  For $\Gamma = -1/2$ there are seven different roots, and when $\Gamma \in (-1/2, +\infty]$, $p(x_{2})$ has twelve roots.  In Fig. \ref{collineal} we can see all the collinear relative equilibria for  $\Gamma_{4} = 1/2$.  

 \begin{figure}[h] 
\centering
      \includegraphics[scale=0.4]{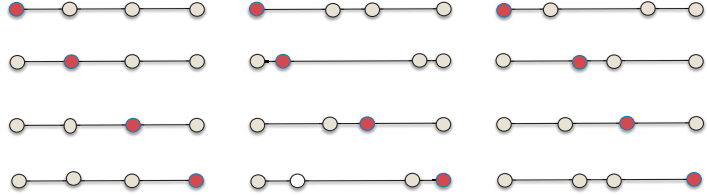}
      \vspace{-5 cm}
\caption{\label{collineal}{\small Set of collinear solutions for $\Gamma_{4} = 1/2$.  Vortices $\Gamma_{1}=\Gamma_{2}=\Gamma_{3}=1$ are detonated by white disks and vortex $\Gamma_{4}$ by dark one.  In the first line we can see the same order for the three configurations but  the  distance among the vortices is different. The same applies to the following three lines}}
\end{figure}

 \section{Symmetric Strictly Planar Relative Equilibria}
 
 In this section we investigate all possible symmetric planar relative equilibria in the four-vortex problem with $\Gamma_1=\Gamma_2=\Gamma_3=1$.  The two possible configurations are a concave kite, and a convex kite.  We use the techniques used in \cite{EJ} for the 4 body problem.
 
 \subsection{The kite family}
 
 We say that a planar relative equilibria has a kite shape if it has an axis of symmetry passing through two of the vorticities.  The kite configuration is \textit{convex} if none of the vortices is located in the interior of the convex hull of the other three, otherwise, if the configuration is not collinear, we say that the kite configuration is \textit{concave}.
 
\begin{figure}[h]
 \centering
 \includegraphics[scale=0.3]{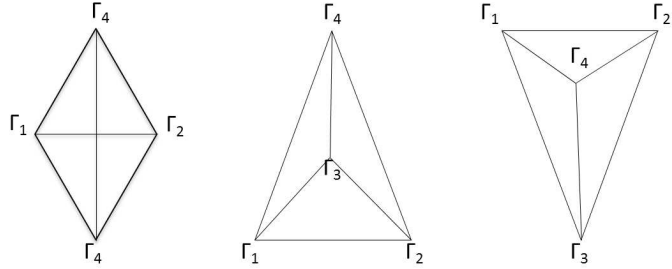}
 \caption{\label{cometa}{\small Kite relative equilibria: (a) $k,l>0$ (left), \; (b) $k<0, l>0$ (center), \;  (c) $k>0, l<0$ (right).}}\label{Kite}
 \end{figure}

In order to study these kite central configurations we choose the axes of coordinates with origin at the center of mass of $\Gamma_3$ and $\Gamma_4$, the $y-$axis as the axis of symmetry, and the $x-$axis orthogonal to it. Taking conveniently the unity of length, we can suppose that the positions of the masses $\Gamma_1, \Gamma_2,
\Gamma_3$ and $\Gamma_4$ are $(-1, 0), (1, 0), (0, -k)$ and $(0, l)$ respectively, and that always $\Gamma_4$ is over $\Gamma_3$ on the $y-$axis; i.e. $k + l > 0$. See Figure \ref{Kite}. 

Using the symmetries of the kite configuration,  Dziobek's equations reduce to the following two equations
\begin{eqnarray}\label{convex}
\Gamma_{4}(k+l)[(1+l^{2})^{-1} - (k+l)^{-2}] + 2k[4^{-1} - (1+k^{2})^{-1}] & = & 0, \nonumber \\
f(k, l) = (k+l)[(1+k^{2})^{-1} - (k+l)^{-2} ] + 2l[4^{-1} - (1+l^{2})^{-1} ] & = & 0,
\end{eqnarray}

\noindent
Solutions $(k, l)$ of equations (\ref{convex}) with $k, l \in \R$, and $k + l > 0$ provide the planar non-collinear relative equilibria of the 4-vortex problem for the vorticities $\Gamma_4 \in \R$ and $\Gamma_{1} = \Gamma_{2} = \Gamma_{3} = 1$.   From the first equation of (\ref{convex}), $k$ cannot be zero; and from the second one $l$ also cannot be zero.

\begin{proposition}\label{equilateral}
The equilateral triangle, with the three vorticities equal to 1
on its vertices and the vorticity $\Gamma_4$ at its barycenter, always is a relative equilibria of the planar 4-vortex problem.
\end{proposition}

\medskip

\begin{proof}  Since the triangle is equilateral, we have that $r_{12} = r_{13} = r_{23} = 2$, so $k = \sqrt{3}$, substituting the $k$ value we have that $l = -1/\sqrt{3}$. \qquad
\end{proof} 

The first equation of (\ref{convex}) can be written as
\[
\Gamma_{4} = \Gamma_{4}(k,l) = \frac{2k[(1+k^{2})^{-1}-4^{-1}]}{(k+l)[(1+l^{2})^{-1}-(k+l)^{-2}]} = \frac{k(3-k^{2})(1+l^{2})(k+l)}{2(1+k^{2})(k^{2}+2kl-1)}.
\]
\noindent
then we can write system (\ref{convex}) as the system

\begin{equation}\label{systems}
 \left\{
      \begin{array}{ll}
      \Gamma_{4} = \Gamma_{4}(k, l), \\
      f(k, l) = 0.
      \end{array}
\right.
\end{equation}

\begin{figure}[h]
\centering
    \includegraphics[scale=0.5]{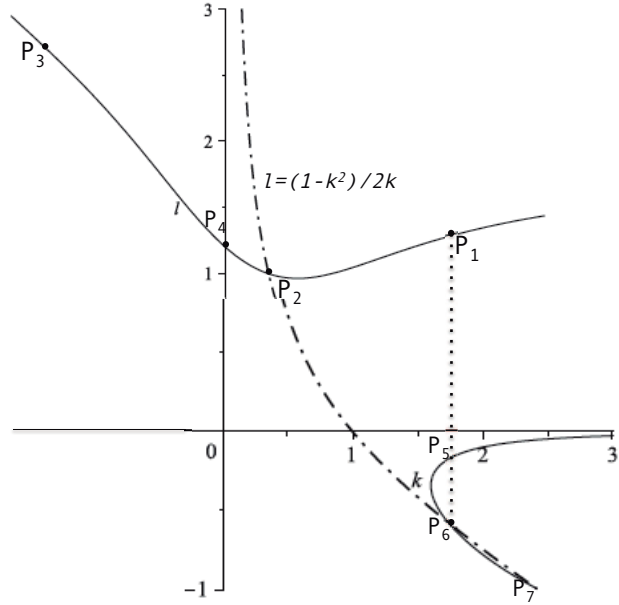}
   \caption{{\small Arcs of $f=0$ where $\Gamma_{4}(k,l)>0$ and $\Gamma_{4}(k, l)<0$ (black curves). The dashed curve is $l=(1-k^2)/2k$ }}\label{fyk}
   \end{figure}

Using \textit{Maple} we plot the curve $f=0$, from here we see that this curve has only two branches in the region $k+l>0$; one contained in the half-plane $l>0$, and the other one contained in the fourth quadrant $\{(k,l):k>0, l<0\}$; see Figure \ref{fyk}.

 Using the expression of $\Gamma_{4}(k,l)$ it follows that
 \[
 \text{sign}(\Gamma_{4}) = \text{sign}\left(\frac{k(3-k^{2})}{k^{2}+2kl-1}\right).
 \]

We note that $k(3-k^{2})$ is positive in $(-\infty, -\sqrt{3})  \cup (0, \sqrt{3})$; is zero in $\{-\sqrt{3}, 0, \sqrt{3}\}$; and negative in the complement.   The intersection of the vertical line $k=\sqrt{3}$ with the curve $f=0$ provides the three points:
\[
 P_{1} =(\sqrt{3}, 1\text{.}19175),\quad P_{5} = (\sqrt{3}, -0\text{.}17633), \quad P_{6} = (\sqrt{3}, -1/\sqrt{3}). 
 \]
\noindent 
  While the intersection of $k=0$ with $f=0$ consists of only the point $P_{4} = (0, 1\text{.}2072)$.  Finally, the intersection of $k = -\sqrt{3}$ with the curve $f=0$ provides a unique point $P_{3} = (-\sqrt{3}, 2\text{.}74748)$ in the region $k+l>0$.
  
  \smallskip
  
  The curve $l(k) = \dfrac{1-k^{2}}{2k}$ only has points with $k>0$ in $k+l>0$, this curve is monotone decreasing  and has a unique branch which intersects the curve  $f=0$ at the points (see Figure \ref{fyk}). 
\begin{eqnarray*}
P_{2} & = & (k(l_{1}), l_{1} = 1), \\
P_{6} & = & (\sqrt{3}, -1/\sqrt{3}), \\
P_{7} & = & (k(l_{2}), l_{2} = -1 ),
\end{eqnarray*}

\noindent
Then, $\Gamma_{4}>0$ if:
\begin{itemize}
\item  $l > -1/\sqrt{3}$ and ($-\sqrt{3} < k < 0$ or $-l + \sqrt{1+l^{2}} < k < \sqrt{3}$), or
\item  $l < -1/\sqrt{3}$ and $\sqrt{3} < k< -l + \sqrt{1+l^{2}}$.
\end{itemize}

\noindent
 And $\Gamma_{4}<0$ if:
 \begin{itemize}
 \item  $l >  -1/\sqrt{3}$ y $0 < k < -l + \sqrt{1+l^{2}}$, y 
 \item  $l < -1/\sqrt{3}$ y $0 < k <\sqrt{3}$.
 \end{itemize}   

\bigskip

\begin{theorem}
The planar 4-vortex problem with three masses equal to 1 and the fourth one equal to $\Gamma_{4} \in \R$ has exactly one convex central configuration which is kite.  Moreover, depending on the value of    $\Gamma_{4}$, it has 1, 2, 3 or 4 kite concave relative equilibria.   

\begin{enumerate}
\item For all $\Gamma_{4} \in \R$ the equilateral triangle with the three equal vorticities located on its vertices and the vorticity $\Gamma_{4}$ located at its barycentre is always a kite concave relative equilibrium.  In the remainder of the statements we omit the description of this concave central configuration.

\item If $\Gamma_{4}=0$ there are exactly 2  additional kite concave relative equilibria configurations.  In one of them, $\Gamma_{4}$ is on a vertex of an isosceles triangle and in the other $\Gamma_{4}$ is in its interior.

\item If $0 < \Gamma_{4} < 1$, there are 3 additional kite concave relative equilibria configurations.  In two of them $\Gamma_4$ is on a vertex of an isosceles triangle and in the other $\Gamma_{4}$ is in its interior.

\item  If $\Gamma_{4} = 1$, then in the convex configuration the four vorticities are located at the vertices of a square, and there is exactly 1 additional kite concave configuration where $\Gamma_{4}$ is on a vertex of an equilateral triangle

\item  If $\Gamma_{4} > 1$, there is exactly 1 additional kite concave configuration having $\Gamma_{4}$ in the interior of an isosceles triangle.

\item  If $\Gamma_{4} < 0$, there is exactly 1 additional kite concave configuration having $\Gamma_{4}$ in the interior of an isosceles triangle. 
\end{enumerate}
\end{theorem}
\begin{proof}  
After intersecting the $\Gamma_4$ function with $f$, we get the following arcs:
\begin{itemize}
\item [(i)] The open arc $\gamma_1$ going from $P_1$ to $P_2$.  We observe that on this arc $k$ and $l$ are positive, so the corresponding relative equilibria associated to the points $(k,l)$ of this arc are convex.  Since $\Gamma_4(k,l)$  takes the value zero on $P_2$ and takes the value $+\infty$ on  $P_1$, there is at least one convex central configuration for every value of $\Gamma_4>0$.

\item [(ii)] The open arc $\gamma_2$ going from $P_2$ to $P_4$.  Since on this arc $k$ and $l$ are positive, the corresponding relative equilibria associated to points $(k,l)$ of this arc are convex.  Since $\Gamma_4$  takes the value zero on $P_4$ and  the value $-\infty$ on $P_2$, there is at least one convex central configuration for every value of $\Gamma_4<0$.

\item [(iii)] The open arc $\gamma_3$ going from $P_3$ to $P_4$.  Since on this arc $k<0$ and $l>0$, the corresponding relative equilibria associated to points $(k,l)$ of this arc are concave having the vorticity $\Gamma_4$ as a vertex of the triangle formed by the convex hull of the four vorticities. Later, we will show that the $\Gamma_4$ function has a unique critical point on this arc.  Therefore, the $\Gamma_4$ value of this point is $1$.  In this arc, when $k$ approaches $0$, $l$ tends to $+\infty$ . So the $\Gamma_4(k,l)$ curve cannot intercept the $f=0$ curve or cut at one or two points.  As will be seen below, when $\Gamma_4=1$, these curves are located in a single point, i.e., there is an equilateral triangle concave configuration.  When $\Gamma_4$ takes values greater than $1$, the two curves do not intersect, and when $\Gamma_4 <1$, the curves  intersect in two points, on those points we have isosceles triangle concave configurations.

\item[(iv)]  The open arc $\gamma_4$ going from $P_5$ to $P_7$.  Since on this arc $k>0$ and $l<0$, the corresponding relative equilibria associated to points $(k,l)$ of this arc are concave having the vorticity $\Gamma_4$ in the  interior of the triangle formed by the convex hull of the other three vorticities.  As $\Gamma_4(k,l)$ takes the value $0$ on $P_5$  and takes the value $+\infty$ on $P_7$ , there is at least one concave configurations for every $\Gamma_4 \in \R$ value.  We remark that $\Gamma_4(k,l)$ on the point $P_6$ takes the value $\Gamma_4=1$ and this point correspond to equilateral triangle of Proposition \ref{equilateral}.
\end{itemize}

We now discuss the case $\Gamma_4=0$, since Dziobek equations are valid for $\Gamma_4 \neq 0$ this case must be studied directly from the equations (\ref{re}) where the center of vorticity is fixed at $c = -k/3$.  So we have that:
 \begin{eqnarray*}
   \lambda(-1, k/3) &=& \frac{(-3-k^{2}, 2k)}{2(1+k^{2})}, \\
   \lambda(1, k/3) &=& \frac{(3+k^{2}, 2k)}{2(1+k^{2})}, \\
   \lambda(0, -2k/3) &=& \frac{(0, -2k)}{1+k^{2}}, \\
   \lambda(0, l+k/3) &=& \frac{(0, 3l^{2}+2kl+1)}{(1+l^{2})(k+l)}.
   \end{eqnarray*} 

\noindent
From there we have three independent equations
\begin{equation*}
   \lambda = \frac{3+k^{2}}{2(1+k^{2})}, \quad \lambda = \frac{3}{1+k^{2}}, \quad \lambda = \frac{3(3l^{2}+2kl+1)}{(3l+k)(1+l^{2})(k+l)}.
   \end{equation*}
\noindent
Solving this system, we obtain the following $(k,l)$ solutions, which satisfy $k+l>0$:
 \[
   (\sqrt{3}, 1\text{.}19175), \; (\sqrt{3}, -0\text{.}176327), \; (\sqrt{3}, -1/\sqrt{3}) \; \text{y} \; (-\sqrt{3}, 2\text{.}74748),
   \]
and correspond to a  convex configuration, one of this is an isosceles triangle with $\Gamma_4$ in the interior of the convex hull formed by the other three vorticities, the other one is the equilateral triangle given in Proposition \ref{equilateral} and the last one is an isosceles triangle configuration with $\Gamma_4$ as a vertex of the triangle. \qquad
\end{proof}

\medskip

Now we will look for the extremals of the  function $\Gamma_ {4} (k, l)$  along the open arc of the curve $f (k, l) = 0$. These points correspond to the bifurcation values ​​ on the number of relative equilibria. For this, we use the method of Lagrange multipliers.

\smallskip

\begin{proposition}
The function $\Gamma_4(k,l)$ restricted to the arcs $\gamma_1 \cup \gamma_2 \cup \gamma_3 \cup \gamma_4$ has a unique critical point on the arc $\gamma_3$.  The value of $\Gamma_4(k,l)$ at  this point is $\Gamma_4=1$.
\end{proposition}

\smallskip 

\begin{proof}  Let $F=\Gamma_4(k,l) + \lambda f(k,l)$, where $\lambda$ is a Lagrange multiplier.  We look for the extremals of the function $\Gamma_4(k,l)$ on the arcs $\gamma_1 \cup \gamma_2 \cup \gamma_3 \cup \gamma_4$. We first look for the extremals of this function on the curve $f(k,l)=0$, and after we must choose the extremals on these arcs.  So, we must solve the system
 \begin{equation}\label{gradiente}
   \frac{\partial F}{\partial k} = 0, \quad \frac{\partial F}{\partial l} = 0, \quad f(k,l) = 0,
   \end{equation}
   of three equations in the three variables $k,l$ and $\lambda$.  First, we eliminate the variable $\lambda$ using the first two equations of (\ref{gradiente}), obtaining a system of the form
   \begin{equation}\label{h}
   h(k,l) = 0, \quad f(k,l)=0,
   \end{equation}
   where 
    \[
   h(k,l) = (1+l^{2})(k+l)^{4}h_{1}(k,l), \qquad f(k,l) = f_{1}(k,l) + 3l^{2}(l^{2}-1)-2, 
   \] 
here $h_1$ is a polynomial in the variables $k$ and $l$ of degrees 9 and 6 respectively, and $f_1$ is a polynomial in the variables $k$ and $l$ of degrees 3 and 4 respectively,
\begin{equation*}
\begin{split}
h_{1}(k,l) & =   6+36\,{l}^{4}-24\,{k}^{2}+9\,{l}^{6}-20\,{k}^{4}+24\,{k}^{6}-18\,{k}^{8}+9\,{l}^{2}+30\,{l}^{5}{k}^{5}\\
\quad &+ 304\,{k}^{4}{l}^{4} +84\,{l}^{3}{k}^{5}-3\,{k}^{9}l-50\,{l}^{6}{k}^{4}+{k}^{9}{l}^{5}+84\,{k}^{6}{l}^{6}+32\,{l}^{7}{k}^{5}+200\,{k}^{7}{l}^{3}\\
\quad & +208\,{k}^{6}{l}^{4}+76\,{k}^{7}{l}^{5}-27\,kl+28\,{k}^{8}{l}^{4}+6\,{l}^{3}{k}^{9}+{k}^{8}{l}^{6}+12\,l{k}^{3}+104\,{l}^{3}{k}^{3}\\
\quad & +9\,k{l}^{5}+126\,{k}^{4}{l}^{2}
+36\,{l}^{6}{k}^{2}+54\,k{l}^{3}+36\,{l}^{2}{k}^{2}+49\,{k}^{8}{l}^{2}+252\,{k}^{3}{l}^{5}\\
\quad & -100\,{k}^{7}l+102\,{k}^{5}l-140\,{k}^{6}{l}^{2} ,
\end{split}
\end{equation*}
and
\begin{equation*}
f_{1}(k,l) = kl(k(l^{2}+3)(k+l)+1+5l^{2}).
\end{equation*}

In order to find the solution of system (\ref{h}), we will use the resultant of two polynomials.  The resultant  Res$[h_1, f_1, k]$ si
\[
6144 l(l^2-3)(l^2+1)^12(3l^2-1)^2 r(l).
\]
with $r$  a polynomial in the variable $l$ of degree 20,
\begin{eqnarray*}
r(l) &=& 36\,{l}^{20}+315\,{l}^{18}-
2457\,{l}^{16}-1776\,{l}^{14}+33264\,{l}^{12}-61986\,{l}^{10}+51534\,{
l}^{8}\\
\quad &-&18904\,{l}^{6}+324\,{l}^{4}+1455\,{l}^{2}+243.
\end{eqnarray*}

\noindent

When equalling to zero the resultant of these two polynomials, the only real roots that we get are $l=\pm \sqrt{3}$ and $\pm1/\sqrt{3}$.  For each real root of $l$, we compute numerically the real roots for $f_1$  and check if the point $(k,l)$ is a solution of the system (\ref{gradiente}) with $\lambda \neq 0$.  After that, we have that the only point that is solution to the system (\ref{gradiente}) on $\gamma_1 \cup \gamma_2 \cup \gamma_3 \cup \gamma_4$ with $\lambda \neq 0$ is  
\[
(k, l) = (-1/\sqrt{3}, \sqrt{3}),
\]
which is equivalent to an equilateral triangle configuration, i.e., at this point we have a bifurcation since the number of relative equilibria changes.  Computing the Hessian of $\Gamma_4(k,l)$ we get that this point is a minimum on the arc $\gamma_3$, and the $\Gamma_4(k, l)$ value at this point is $1$.  \qquad
\end{proof}

\bigskip

Counting the different positions for the three equal vortices in configuration of relative equilibria described previously and adding the 4 classes of  collinear relative equilibria we obtain the number of planar relative equilibria in kite configuration. The result is summarizing as:

\medskip

\begin{corollary}
The planar 4-vortex problem with three vorticities equal to 1 and the fourth one equal to $\Gamma_4$ has the following classes of kite relative equilibria:
\begin{itemize}
\item[1.] 26 \, if \, $\Gamma_4=0$, 
\item[2.] 29 \, if \, $0<\Gamma_4<1$,
\item[3.] 34 \, if \, $\Gamma_4=1$,
\item[4.] 23 \, if \, $\Gamma_4>1$,
\item[5.] 20 \, if \, $-1/2<\Gamma_4<0$,
\item[6.] 15 \, if \, $\Gamma_4 = -1/2$,
\item[7.] 14 \, if \, $-1 <\Gamma_4<-1/2$,
\item[8.] 8 \, if \, $\Gamma_4 < -1$.
\end{itemize}
\end{corollary}

\begin{proof}
The idea is to count the number and type of solutions (equivalence classes) for different values of $\Gamma_4$.

When all vorticities are equals  ($\Gamma_4 =1$) we have $6$ square configurations, $8$ configurations of equilateral triangle where $\Gamma_4$ is an interior vortex (that is $\Gamma_4$ is in the interior of the convex hull of the other three), 8 configurations of equilateral triangle where $\Gamma_4$ is an exterior vortex (it is located at one vortex of the equilateral triangle) and 12 collinear solutions.

\noindent
If $\Gamma_4 = 0$, we have 12  collinear configurations, 6 convex configurations, 2 configurations of equilateral triangle with interior vortex, 3 isosceles configurations with interior vortex and 3 isosceles configurations with exterior vortex.

\noindent
If $0 < \Gamma_4 < 1$ we have 12 collinear configurations, 6 convex configurations, 2  configurations of equilateral triangle with the different vorticity inside, 6 relative equilibria where $\Gamma_4$ is on a vertex of an isosceles triangle and 3 relative equilibria where $\Gamma_4$ is in the interior of an isosceles triangle.

\noindent
If $\Gamma_4>1$ we have 12 collinear configurations, 6 convex configurations, 2  configurations of equilateral triangle with the different vorticity inside, and 3 relative equilibria where $\Gamma_4$ is in the interior of an isosceles triangle.

\noindent
If $-1/2 < \Gamma_4 < 0$, we have  12 collinear configurations, 6 convex configurations, and 2  configurations of equilateral triangle with the different vorticity inside.

\noindent
If $\Gamma_4=-1/2$, there are 7 collinear relative equilibria, 6 convex relative equilibria, and 3 relative equilibria of equilateral triangle.

\noindent 
If $-1 < \Gamma_4 < -1/2$, there are 6 collinear relative equilibria,  6 convex relative equilibria, and 3 relative equilibria of equilateral triangle.  \qquad
\end{proof}

\subsubsection{Rhombus configuration}

Now, we will study a special case of convex kite configuration, the rhombus.  The special thing about this configuration is that, because of the two lines of symmetry, we can find an explicit expression for the value of the sides which depend on different vorticities.

One of the properties of the rhombus is that all four sides are congruent, so $r_{13} = r_{14} = r_{23} = r_{24}$ where the diagonals satisfy the relation $4r_{13}^2 = r_{12}^2+r_{34}^2$.  We fix the areas orientation as $A_1=A_2=-A_3=-A_4$ (since $A_1+A_2+A_3+A_4$).   Let $x=r_{34}/r_{12}$ be the ratio between the diagonals of the rhombus.  From the first and second equation of (\ref{Dziobek}) we have that $x^2=1$, i.e., the diagonal are equal, so the configuration is a square with all the vorticities equal.  From the third and fifth equation in (\ref{Dziobek}) we have an expression for $\Gamma_4$ given by
\begin{equation}\label{rombo}
\Gamma_4 = \frac{x^2-3}{1-3x^2}.
\end{equation}
$\Gamma_4$ is positive when $1/\sqrt{3} < x < \sqrt{3}$, and negative when $x>\sqrt{3}$ or $0<x<1/\sqrt{3}$  (see Fig. \ref{fyg}).  Solving  equation (\ref{rombo}) for $x^2$, yields two different rhombus families.  One when $\Gamma_4 \in (-1/3, +\infty)$, and the other one when $\Gamma_4<3$.  This can be seen by inverting equation (\ref{rombo}), which yields
\begin{equation}\label{35}
x^{2} = \frac{\Gamma_{4}+3}{3\Gamma_{4} + 1}.
\end{equation}

\begin{figure}
\centering
\includegraphics[scale=0.4]{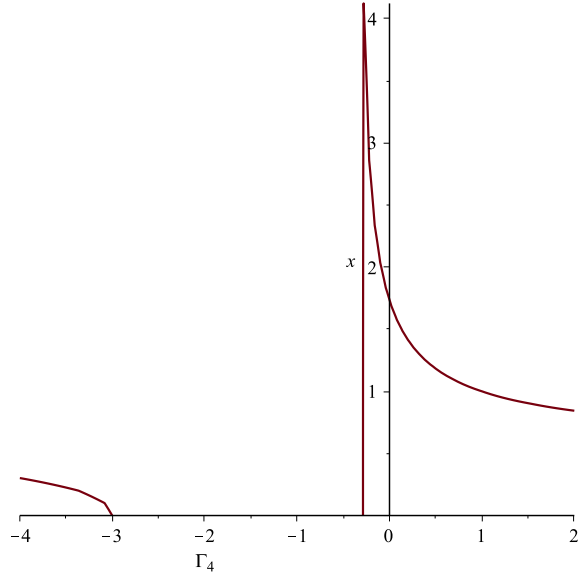}
      \caption{{\small Plot of $x$ vs $\Gamma_{4}$, where $x=r_{34}/r_{12}$ is the ratio between the diagonals of the rhombus}}\label{fyg}
   \end{figure}
   
Since $\lambda = \dfrac{-L}{2I}$, we can calculate the angular velocity ,
\begin{equation}\label{36}
\lambda =- \frac{3(1+\Gamma_{4})}{r_{12}^{2}},
\end{equation}

\noindent
The numerator vanishes at $\Gamma_{4} =  -1$, in this case $x^{2} < 0$, i.e., it is not possible to have absolute equilibria when the vortices are in a rhombus configuration.  The value of the angular velocity $\lambda$ is negative when $\Gamma_{4}  > -1$, and $\lambda$  is positive when $\Gamma_{4} < -1$.  We summarize our conclusions in the following theorem:

\medskip

\begin{theorem}
There are two one-parameter families of rhombus relative equilibria with vortex strengths  $\Gamma_{1} = \Gamma_{2} = \Gamma_{3} = 1$.  The vortices 1 and 2 lie on opposite sides of each other, as do vortices  3 and 4.  The mutual distances are given by

\begin{equation}
\left(\frac{r_{34}}{r_{12}}\right)^{2} = \frac{\Gamma_{4} + 3}{3\Gamma_{4}+1} \qquad \text{and} \qquad \left(\frac{r_{13}}{r_{12}}\right)^{2} = \Gamma_{4} + 1,
\end{equation}
describing two distinct solutions.  For $\Gamma_{4} \in(-\infty, -3)$, we have a solution that has $\lambda >0$.  The other solution is when $\Gamma_{4} \in (-1/3, 0)$ that has $\lambda < 0$.  The case $\Gamma_{4} = 1$ reduces to square.  For $\Gamma_{4} >0$, the larger vortex lies on the shorter diagonal.
\end{theorem}

\begin{proof}  The formula for $(r_{13}/r_{12})^{2}$ comes from $1+x^{2} = 4(r_{13}/r_{12})^{2}$.
For the case $\Gamma_{4} > 0$ we get from equation (\ref{rombo}) that $\Gamma_{4} < 1$ if and only if $1 < x < \sqrt{3}$.  Beginning with the square at $\Gamma_{4} = 1$, as $\Gamma_{4} \in (0,1)$ , the ratio of the diagonals of the rhombus increases from $1$ to $\sqrt{3}$, so the different vorticity ($\Gamma_{4}$) is located on the longer diagonal, while if $\Gamma_{4}>1$, the ratio of the diagonals of the rhombus is always less than $1$ and decreases as  $\Gamma_{4}$ increases. In this case $\Gamma_{4}$ is located on the shorter diagonal.

At $\Gamma_{4} = -3$ a bifurcation occurs, and a new family is born emerging out of a binary collision between vortices $3$ and $4$.  This family has the vortex  $\Gamma_{4}$ located on the shorter diagonal. \qquad
\end{proof}
\section*{Acknowledgments}
This work has been partially supported by CONACYT-M\'exico, grant 128790 and by a NSERC Discovery grant.


\end{document}